\documentclass[runningheads,a4paper]{llncs}

\usepackage{breakcites}
\usepackage{amsmath}
\usepackage{amssymb}
\usepackage{graphicx}

\usepackage{algorithm}
\usepackage{algorithmic}

\usepackage{listings}
\usepackage{xcolor}

\usepackage{amsfonts}
\usepackage{float}
\usepackage{enumitem}

\floatstyle{ruled}
\newfloat{Program}{ht}{lop}

\newtheorem{thm}{Theorem}
\newtheorem{lmm}{Lemma}

\lstdefinelanguage{scala}{
  morekeywords={abstract,case,catch,class,def,%
    do,else,extends,false,final,finally,%
    for,if,implicit,import,match,mixin,%
    new,null,object,override,package,%
    private,protected,requires,return,sealed,%
    super,this,throw,trait,true,try,%
    type,val,var,while,with,yield},
  otherkeywords={=>,<-,<\%,<:,>:,\#,@},
  sensitive=true,
  morecomment=[l]{//},
  morecomment=[n]{/*}{*/},
  morestring=[b]",
  morestring=[b]',
  morestring=[b]"""
}

\urldef{\mailsa}\path|{dtolpin}@paypal.com|
\newcommand{\keywords}[1]{\par\addvspace\baselineskip
\noindent\keywordname\enspace\ignorespaces#1}

\lstdefinelanguage{scala}{
  morekeywords={abstract,case,catch,class,def,%
    do,else,extends,false,final,finally,%
    for,if,implicit,import,match,mixin,%
    new,null,object,override,package,%
    private,protected,requires,return,sealed,%
    super,this,throw,trait,true,try,%
    type,val,var,while,with,yield},
  otherkeywords={=>,<-,<\%,<:,>:,\#,@},
  sensitive=true,
  morecomment=[l]{//},
  morecomment=[n]{/*}{*/},
  morestring=[b]",
  morestring=[b]',
  morestring=[b]"""
}

\begin{document}

\mainmatter  

\title{Progressive Temporal Window Widening}

\titlerunning{Progressive Temporal Window Widening}

%
%
\author{David Tolpin}


\institute{PayPal, \email{dtolpin@paypal.com}}

%
%

\toctitle{Lecture Notes in Computer Science}
\tocauthor{Authors' Instructions}
\maketitle

\begin{abstract}
	This paper introduces a scheme for data stream processing
	which is robust to batch duration. Streaming frameworks
	process streams in batches retrieved at fixed time
	intervals. In a common setting a pattern recognition
	algorithm is applied independently to each batch. Choosing
	the right time interval is tough --- a pattern may not fit
	in an interval which is too short, but detection will be
	delayed and memory may be exhausted if the interval is too
	long.  We propose here Progressive Window Widening, an
	algorithm for increasing the interval gradually so that
	patterns are caught at any pace without unnecessary delays
	or memory overflow.

	This algorithm is relevant to computer security, system
	monitoring, user behavior tracking, and other applications
	where patterns of unknown or varying duration must be
	recognized online in data streams. Modern data stream
	processing frameworks are ubiquitously used to process high
	volumes of data, and adaptive memory and CPU allocation,
	facilitated by Progressive Window Widening, is crucial for
	their performance.
\end{abstract}


\section{Introduction}

We consider here the problem of windowed data stream
processing~\cite{GO03}.  A data stream is a real-time,
continuous, ordered sequence of items.  In the windowed setting,
the arriving data are divided into windows, either by time
interval or by data size, and a pattern recognition algorithm,
based on a data mining or machine learning approach, is applied
to each window to discover exact or approximate patterns
appearing in the window~\cite{G12}. Here, we view a pattern
recognition algorithm as a black box function on stream
fragments. For example, a pattern can be an episode --- a
partially ordered sparse subsequence~\cite{MTI97}, the language
of the text, or the most likely goal of the sequence of actions
in the fragment.

Windowed data stream processing is frequently used in computer
security~\cite{WFP99,VJ07,YDK14}, user behavior
tracking~\cite{ABA+14}, sensor data analysis for system
monitoring~\cite{AFN+07}, and other applications. The right
choice of window size is crucial for efficient data processing
and timely response. Data are divided either into physical
windows, by time interval, or into logical, or count-based,
windows, by data size or number of records in a single
window~\cite{GO03,G12}. 

The choice of either physical or logical
windows depends both on properties of the data stream and on the
objective of the data processing algorithm. Logical
windows are more naturally handled by machine learning
algorithms with inputs of fixed size~\cite{G12}, while
physical windows allow both more efficient
processing and faster online response~\cite{HYZ+10,ZDL+12,ZDL+13}.
This paper explores selecting a window size for
physical, interval-based windows. The dilemma behind selecting
a window size which inspired this research is 
\begin{itemize}
    \item  whether to choose a
smaller window and sacrifice context, such that no single window
contains a complete pattern,
\item or to increase the window size at the cost of increased
    consumption of computational resources and delayed response.
\end{itemize}
\begin{figure}
    \centering
	\includegraphics[scale=0.5]{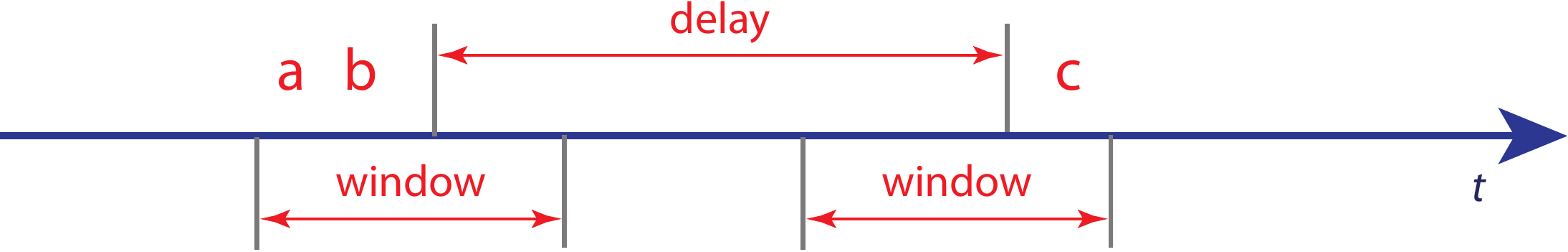}
    \caption{Adversary escaping detection. Pattern $a, b, c$
    cannot be caught because the delay between $b$
    and $c$ is longer than the window duration.}
    \label{fig:intruder}
\end{figure}
This dilemma is relevant to many applications of data stream
processing, but in particular to security
applications~\cite{WFP99,VJ07,YDK14}, where an adversary
aware of the maximum window time interval can escape the
detection algorithm by introducing delays between data stream
entries (such as transactions or web site accesses) which exceed
the interval and prevent detection (Figure~\ref{fig:intruder}).
Even if the maximum duration of a pattern is known in advance,
setting the window size to exceed the maximum duration means
that recognition of any shorter pattern will be delayed.

To address this dilemma, we introduce an algorithm which we call
\textit{Progressive Window Widening} (PWW). PWW processes the
data stream through an array of sliding windows of increasing
physical size, such that shorter patterns are recognized sooner,
however windows covering longer patterns are also applied to the
stream.  Despite employing several window sizes in parallel, PWW
still remains efficient in CPU and memory consumption. The paper
proceeds as follows: first, necessary preliminaries are
introduced in Section~\ref{sec:preliminaries}. Then, the
algorithm is described and analysed (Sections~\ref{sec:pww}
and~\ref{sec:analysis}), as well as evaluated empirically
(Section~\ref{sec:case-studies}).  Finally, related work is
reviewed, and contribution and future research are discussed
(Sections~\ref{sec:related} and~\ref{sec:contribution}).

\section{Preliminaries}
\label{sec:preliminaries}

\subsection{Batched Stream Processing}

In \textit{batched stream processing}, which we adopt in this
paper as a lower level for PWW, stream data arrives in batches
--- sequences of fixed duration. Several batches can be combined
into a window of size equal to the total size of the batches
composing the window. Along with batch size (or duration, used
interchangeably here), a batch is characterized by its
\textit{length}, the number of atomic elements, or records, it
contains. For example, a one-minute batch of web site log stream
may contain 1000 entries --- we shall say that the size, or
duration of the batch is 1 minute, and the length of the batch
is 1000 entries.

Further on, we extend the note of batched stream processing by
stating that a data stream with batch duration $t$ may be
transformed into a data stream with batch duration $kt$ by
concatenating each $k$ consecutive batches together. Denoting
a batch of the original stream with batch duration $t$ by
$B_{i,l}$
and a batch of the combined stream with batch duration $kt$ by
$B_{i+1,j}$ for some $i$, $j$, and $l$, one may write ($\circ$ stands for batch
concatenation):

\begin{equation}
    B_{i+1,j} = B_{i,kj-k+1}\circ B_{i,kj-k+2}\circ \dots \circ B_{i,kj}\quad \forall j \in \mathbb{N}^+
    \label{eq:combined-batch}
\end{equation}

\subsection{Sliding Windows}

Depending on the overlay between windows, one discerns between
\textit{tumbling} (there are gaps between windows),
\textit{jumping} (the windows are adjacent), and
\textit{sliding} (overlapping) windows~\cite{GO03}. PWW is based on sliding
windows; the next window starts earlier than the current window
terminates. 

Sliding windows have several uses. We are interested in one
particular case: sliding windows with a half-size overlap; the
feature we are interested in is described by Lemma~\ref{lmm:sliding}:
\begin{figure}
    \centering
	\includegraphics[scale=0.55]{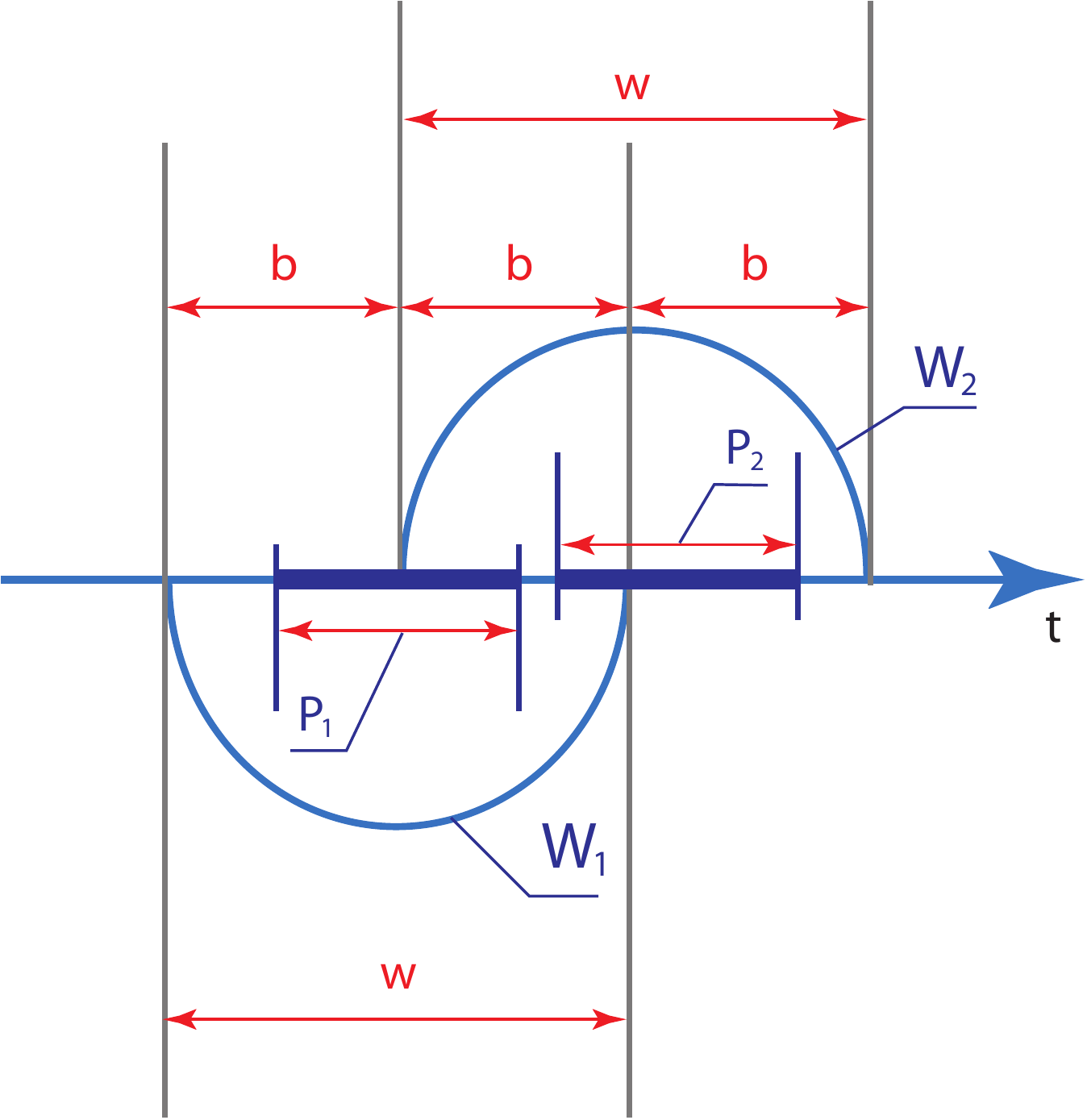}
	\caption{Pattern recognition in sliding windows. If we want to
recognize patterns of duration at most $t$, it is sufficient to use
sliding windows of size $2t$ with half-size overlap.}
    \label{fig:sliding-window}
\end{figure}
\begin{lmm}
    \label{lmm:sliding}
A sequence of sliding windows of size $2b$ with overlap $b$
covers any interval of size at most $b$.
\end{lmm}
\begin{proof}
    Indeed, divide the stream into batches of size $b$
    (Figure~\ref{fig:sliding-window}). Any
    interval of size at most  $b$ is either entirely within a
    single batch, or spans two consequent batches. But every
    single batch, and every pair of consequent batches is
    covered by a single window. This completes the proof.
\end{proof}
A corollary of Lemma~\ref{lmm:sliding} is that if we want to
recognize patterns of duration at most $t$, it is sufficient to use
sliding windows of size $2t$ with half-size overlap.

\section{Progressive Window Widening}
\label{sec:pww}

We introduce here \textit{Progressive Window Widening,} an
algorithm for progressive widening of temporal windows. To
define the algorithm efficiently, we rely on two auxiliary
notions:
\begin{itemize}
	\item $L_{max}$ --- the maximum length of a data sequence which 
		may contain a pattern. For example, if a game player must
		complete each game round in 20 moves, than any pattern
		pertaining to a single round must be contained within 20
		moves. Alternatively, $L_{max}$ can be chosen such that
		the probability of a random occurence of the pattern
		in a data sequence of length $L_{max}$ is sufficiently
		low \cite{GAS03}.
    \item $T_{max}$ --- the upper bound on pattern duration.
        For example, if a computer is rebooted every week, then
        the longest duration of a running process is one week,
        or $604\,800$ (less than $2^{20}$) seconds. $T_{max}$ 
        is not strictly required for the definition of the
        algorithm but helps in the algorithm's implementation.
\end{itemize}
The algorithm processes the data stream in parallel, through
multiple asynchronous sliding windows of different sizes.

\subsection{Algorithm Outline}

\begin{algorithm}
    \begin{algorithmic}[1]
    \STATE \textbf{procedure} \textsc{PWW}($S$ -{}- stream, $t$ -{}- batch duration)
    \STATE \textsc{Sleep}($t$) \label{alg:pww-sleep}
    \STATE Create stream $S'$ from $S$ with batch duration
    $2t$  (see Algorithm~\ref{alg:dropping}) \label{alg:pww-double}
    \STATE Call \textsc{PWW}($S'$, $2t$) \textit{asynchronously} \label{alg:pww-recursive}
    \FOR {\textbf{each} sliding window $W$ \textbf{in} $S$} \label{alg:pww-sliding}
        \IF {patterns present in $W$} \label{alg:pww-detect}
            \STATE Output detected patterns 
        \ENDIF
    \ENDFOR
\end{algorithmic}
\caption{Progressive Window Widening}
\label{alg:pww}
\end{algorithm}

PWW (Algorithm~\ref{alg:pww}) performs the following operations:
\begin{enumerate}
	\item Recursively combines pairs of adjacent batches, doubling batch
        duration of each stream and creating a stream with
        batches of double duration (line~\ref{alg:pww-double}).
	\item Runs a detection algorithm in a sliding window on each
        stream (line~\ref{alg:pww-detect}).
    \item While combining batches, discards subintervals of
        combined batches which cannot intersect a yet unseen pattern 
        (see Section~\ref{sec:combining-batches} for detailed
        explanation).
\end{enumerate}
As the algorithm runs, multiple batched streams are created, and
sliding windows move through each of the
streams~(Figure~\ref{fig:window-widening}). The algorithm relies
on asynchronous recursive calls to \textsc{PWW}
(line~\ref{alg:pww-recursive}). Asynchronous
calls are possible because the processing of each stream is
independent. Such asynchronous execution is particularly
suitable for modern multi-core multi-node cluster architectures:
different invocations of \textsc{PWW} may be executed on
different cores or different nodes in the cluster. 

Note that extra streams are created (lines~\ref{alg:pww-sleep}--\ref{alg:pww-recursive}) and processed
(line~\ref{alg:pww-sliding}) with exponentially increasing delays,
since a window can be processed only upon termination of the window's interval.
\begin{figure}
	\centering
	\includegraphics[scale=0.55]{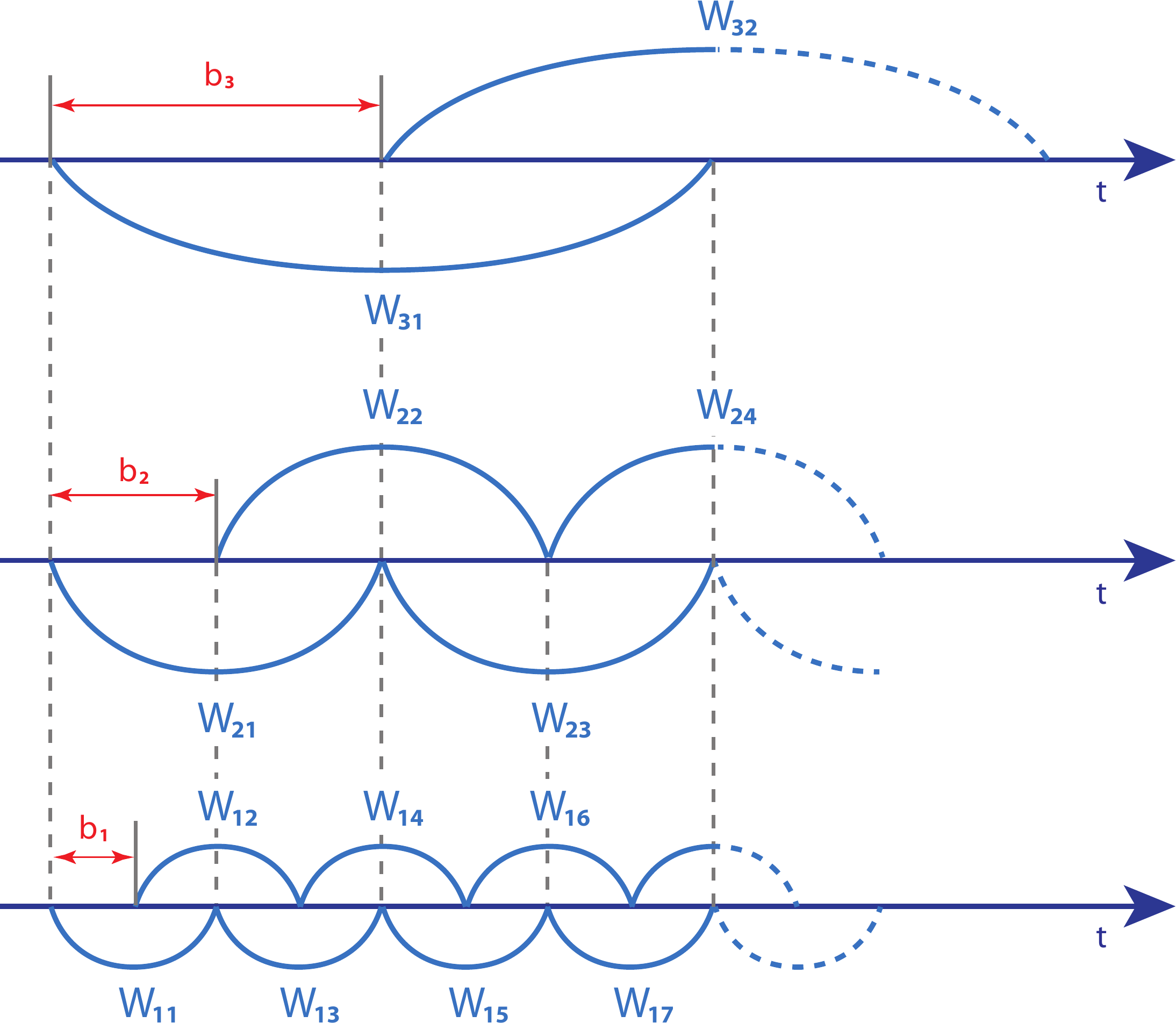}
	\caption{Window widening. Multiple sliding windows of increasing
    durations are processed in parallel, with delay.}
    \label{fig:window-widening}
\end{figure}

\subsection{Combining batches}
\label{sec:combining-batches}

An integral part of PWW is the optional discarding of a subinterval while
combining two subsequent batches.  For every stream of batches of
duration $t$, the algorithm waits $2t$ time units for 2 batches to
arrive.  Then, a stream of base duration $2t$ is formed by combining
the batches (Algorithm~\ref{alg:dropping}).
\begin{algorithm}
    \begin{algorithmic}[1]
    \STATE \textbf{procedure} \textsc{Combine}($B_{i-1,2j-1}$,
        $B_{i-1,2j}$ -{}- consequent batches)
    \STATE $B_{i, j}$ $\gets$ \textsc{Concatenate}($B_{i-1,2j-1}$, $B_{i-1,2j}$) \label{alg:dropping-concatenate}
    \IF {\textsc{Length}($B_{i, j}$) $> 2L_{max}$} \label{alg:dropping-start}
       \STATE \textsc{Remove}($B_{i, j}$, from=$L_{max}$, till=\textsc{Length}($B_{i, j}$)$- L_{max}$)
    \ENDIF \label{alg:dropping-end}
	\RETURN $B_{i, j}$
\end{algorithmic}
\caption{Combining Batches}
\label{alg:dropping}
\end{algorithm}
PWW combines batches by concatenation
(line~\ref{alg:dropping-concatenate}). If the length of the combined
batch is greater than $2L_{max}$, the middle part of the combined
batch is discarded (Figure~\ref{fig:dropping}), leaving subsequences
of length $L_{max}$ at both ends of the batch
(lines~\ref{alg:dropping-start}--\ref{alg:dropping-end}).
\begin{figure}
    \centering
	\includegraphics[scale=0.55]{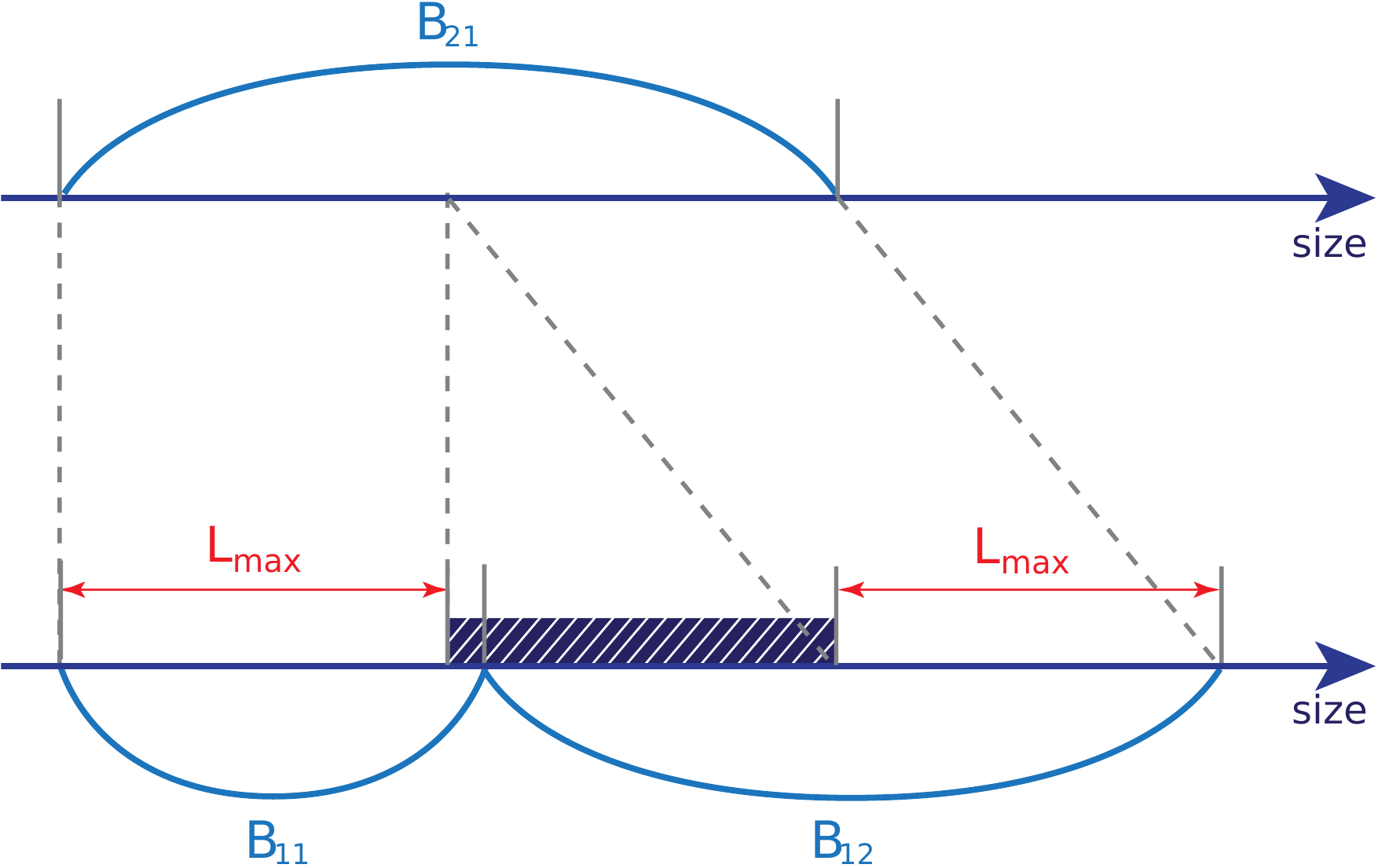}
	\caption{Removing useless data from batches. Since no
    pattern can span a fragment of more than $L_{max}$ items,
    the middle part of the combined batch can be discarded.
    Only subsequences of length $L_{max}$ at the ends of the
    batch must be kept.}
    \label{fig:dropping}
\end{figure}
Consequently, no batch in any stream is longer than $2L_{max}$. The
subintervals may be discarded because a combined batch at the next level
coincides with a sliding window at the current level, so new patterns may be
discovered only between batches, rather than within a single batch.

\section{Algorithm Analysis}
\label{sec:analysis}

In this section we show that the algorithm eventually has a chance to
detect a pattern of any duration and, at the same time, runs in
bounded resources.

\subsection{Correctness}

Since window duration is unbounded, to prove the correctness we just
need to show that discarded intervals do not intersect any
pattern which did not fall entirely within a single window.

\begin{thm}Any pattern of length at most $L_{max}$ is contained in 
    a window.
\end{thm}
\begin{proof}Indeed, as we noted earlier, a combined batch at the next
    level coincides with a single sliding window at the current level.
    Any pattern which is contained in a sliding window could be
    seen by the pattern recognition algorithm, and the interval
    containing the pattern can be discarded. Hence, a yet unseen
    pattern intersecting a window must cross one of the ends of the
    window (and of the combined batch at the next level). 

    Since every combined batch with a discarded subinterval has
    $L_{max}$ contiguous elements adjacent to each of the ends, 
    the discarded interval does not intersect with a pattern of length
    at most $L_{max}$. This completes the proof.
\end{proof}

\subsection{Complexity}

We launch an unbounded number of parallel processes, and want to show
that PWW runs in computationally bounded resources. The work that
the algorithm performs is assumed to take place inside a pattern
recognition algorithm run on each sliding window. Let us denote
the resources (a combination of memory and amount of work)
required to run a certain pattern recognition algorithm on
window of length $l$ by $\mathcal{R}(l)$. Then, the following
theorem holds:
\begin{thm}
    \label{thm:complexity}
    Denote by $t$ the batch duration of the initial, uncombined
    stream. Assume that the maximum length of a batch of the initial
    stream does not exceed $2L_{max}$. Then the average resources
    $\rho$ per unit time required to run PWW are bounded by a constant:
    \begin{equation}
        \rho \le \frac {2\mathcal{R}(4L_{max})} t\,.
		\label{eqn:rho}
    \end{equation}
\end{thm}
\begin{proof}
    Due to Algorithm~\ref{alg:dropping}, the length of a sliding window
	is at most $4L_{max}$, hence running the pattern recognition
    algorithm on a window requires at most $\mathcal{R}(4L_{max})$
    resources.

    Windows in streams are processed sequentially, and a window in the
	$i$th stream arrives after delay $2^{i-1}t$. Therefore, 
    \begin{equation}
        \rho \le \sum_{i=1}^{\infty} \frac {\mathcal{R}(4L_{max})} {2^it}
			   = \frac {\mathcal{R}(4L_{max})} t \sum_{i=1}^\infty \frac 1 {2^{i-1}}
               = \frac {2\mathcal{R}(4L_{max})} t\,.
    \end{equation}
    This completes the proof.
\end{proof}
Note that the assumption in Theorem~\ref{thm:complexity} is 
satisfied by choosing the initial batch duration $t$ to be small
enough. On the other hand, it may be the case that the length
of a batch at any intermediate level reaches $2L_{max}$ (and
then the data in the batch is partially disregarged, as detailed
in Section~\ref{sec:combining-batches}.

In practice, the number of parallel streams may be bounded.
However, even if unbounded, average resources required to run
the algorithm are bounded.

\section{Case Study: Detecting Remote Shells in a System Call Stream}
\label{sec:case-studies}

In this case study, we monitor an online stream of system calls 
from a network-connected server, and want to detect possible
invocations of remote shells as soon as possible.
System call traces are represented according to the following format:
\begin{verbatim}
system-call [argument ...] [=> return-value]
\end{verbatim}
A line consists of the system call name, followed by optional
arguments, followed by optional return value preceded by
\texttt{=>}.  Each argument is a name-value pair, with the name
separated from the value by \texttt{=}.  System call sequences
corresponding to remote shell invocations can be interspersed
with unrelated activities.

For simplicity, we limit detection to a single episode which may
correspond to accepting a network connection and then launching
a shell communicating with the remote user through the connection:
\begin{verbatim}
1 accept fd=x => y
2 dup fd=y => 0 | dup fd=y => 1 | dup fd=y => 2
3 execve exe=z
\end{verbatim}
In the above pseudocode, system call name is followed by
\texttt{name=value} argument pairs and then by return value
preceded by \texttt{=>}.  In a matching system call sequence
\texttt{y} must have the same value in lines \texttt{1} and
\texttt{2}, three system calls in line \texttt{2} may be
executed in any order, and \texttt{x}, \texttt{z} may take any
value.  For example, sequence
\begin{verbatim}
accept fd=5 => 6
dup fd=6 => 2
dup fd=6 => 1
dup fd=6 => 0
execve exe=sh
\end{verbatim}
matches the episode.

For the empirical evaluation we use a sequential version of PWW
which facilitates easy estimation of the amount of work. We set
$L_{max}=100$ because malicious code is often transmitted in a
single packet with only a few dozens of instructions. For
simplicity, we assume that  one system call arrives per time
unit. We use a stream of $10\,000$ system calls recorded on a
Linux machine, into which we inject episode instances with
varying delays between instructions. As a baseline, we use a
fixed duration window of $200$ time units. We find that:
\begin{figure}[t]
	\centering
    \includegraphics[scale=0.65]{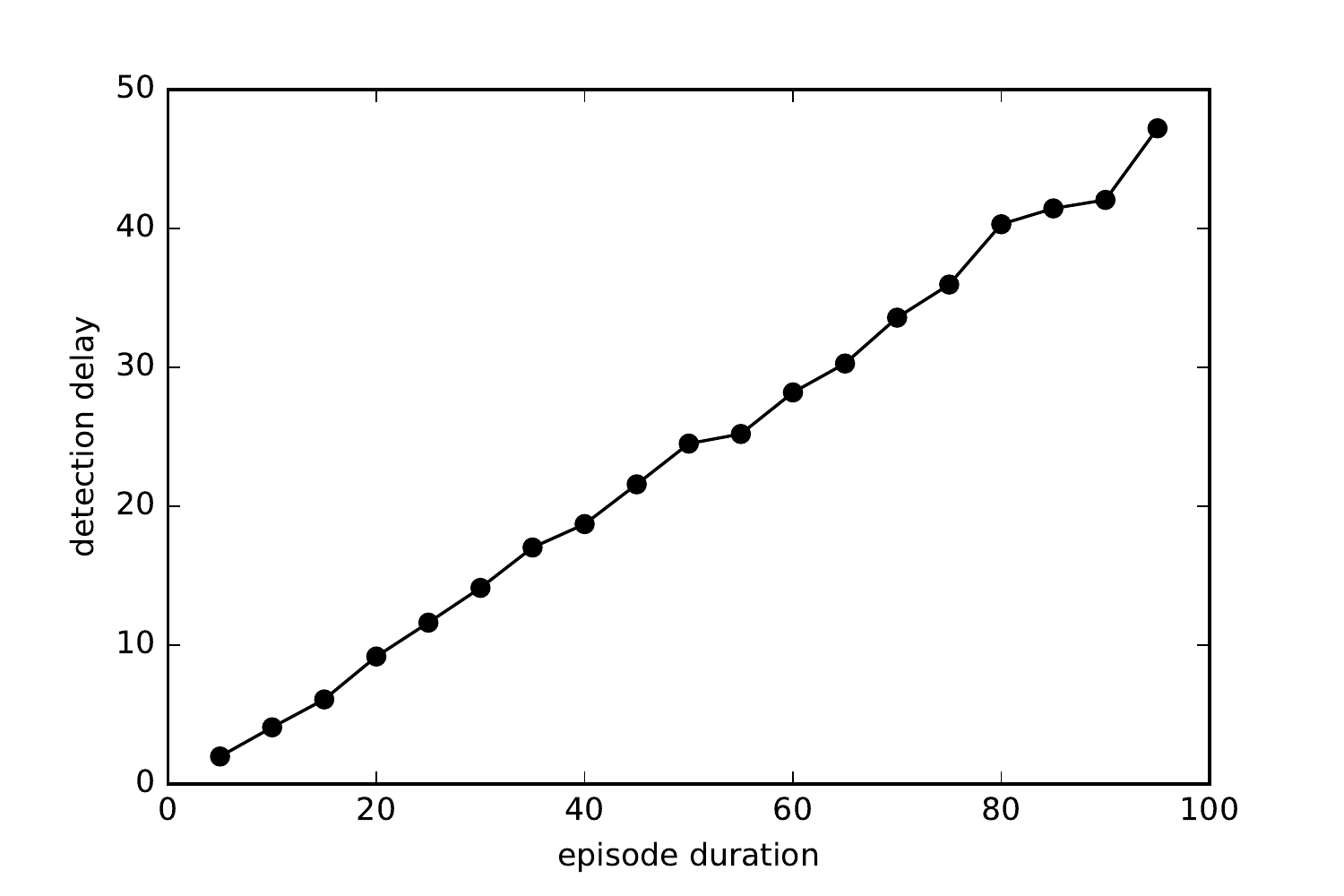}
	\caption{Detection delay. The delay grows linearly
    with shell code duration, with factor 0.5, as expected. In
    other words, by linearly increasing the amount of work we
    are able to detect patterns with delay which is only half of
    the pattern duration.}
    \label{fig:detection-delay}
\end{figure}
\begin{figure}[t]
	\centering
    \includegraphics[scale=0.65]{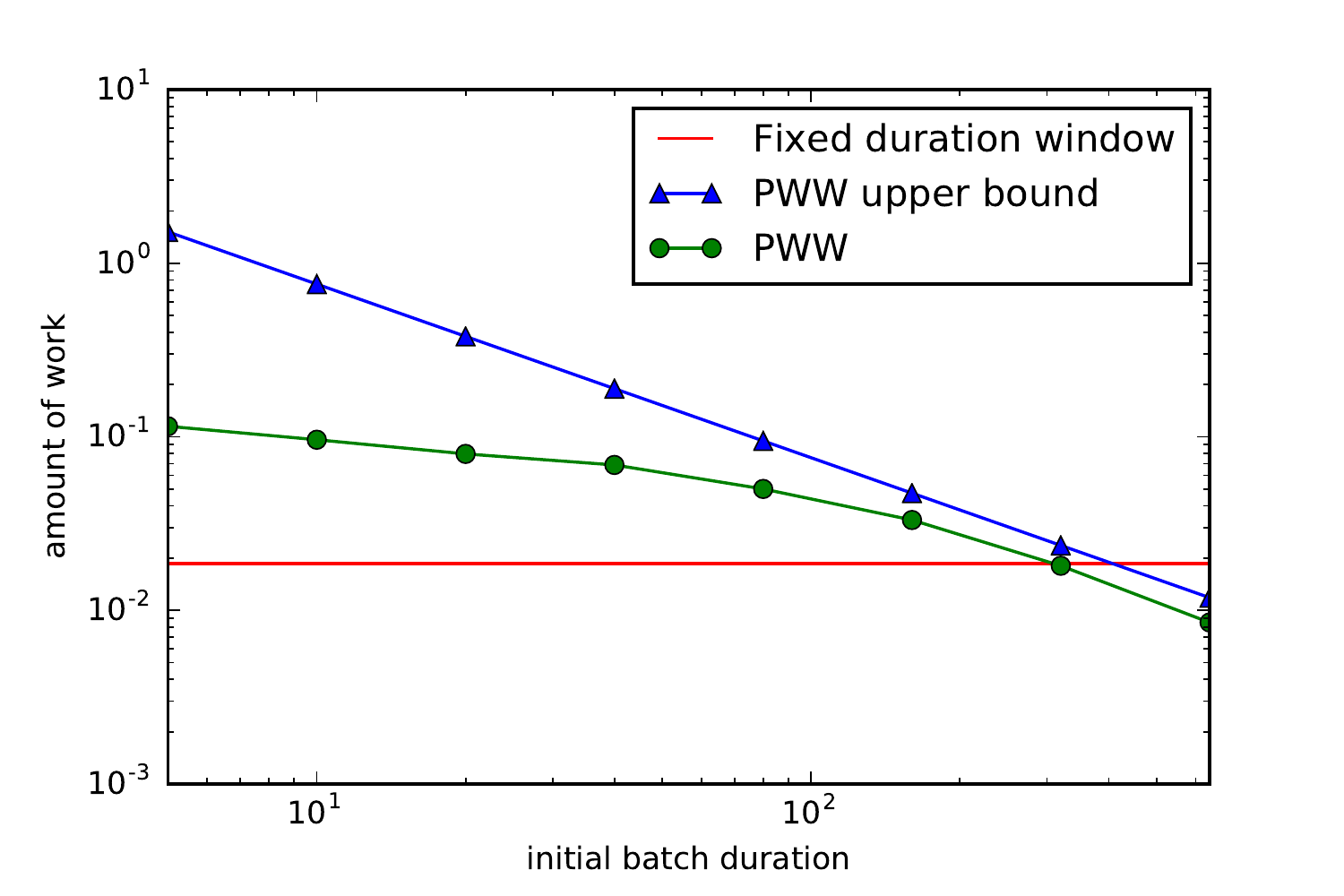}
    \caption{Amount of work. The amount of work of PWW (green)
    approaches but stays below the theoretical bound (blue).
    The amount of work of PWW is lower than of fixed duration window
    (red) for sufficiently large initial batch duration.}
    \label{fig:amount-of-work}
\end{figure}
\begin{itemize}
	\item The detection delay (Figure~\ref{fig:detection-delay})
		is proportional to the episode duration with factor 0.5.
	\item The amount of work (Figure~\ref{fig:amount-of-work})
        approaches but stays below bound (\ref{eqn:rho}) as the
        initial batch duration grows. For sufficiently large
        initial batch duration, the amount of work of PWW is
        lower than of fixed window duration.
\end{itemize}
The results are in accordance with the algorithm analysis.  If
a fixed duration window were used, either the average detection
delay would be larger, or some episodes were left undetected.
PWW ensures timely detection of episodes of any duration at the
cost of only a constant factor increase in the amount of work.

The source code, data, and results for the case study are available
at \url{https://bitbucket.org/dtolpin/pww-paper-case-studies}.
The evaluation notebook can be viewed in the browser at
\url{http://tinyurl.com/jgknulz}.

\section{Related Work}
\label{sec:related}

While Progressive Window Widening can be implemented from
scratch on low-level data streams, the algorithm was inspired 
and relies in implementation on batched stream processing. Batch
stream processing was introduced in Comet~\cite{HYZ+10}. Apache
Spark offers Spark Streaming~\cite{ZDL+12,ZDL+13}, a powerful
implementation of programming model \textit{discretized
streams}. Discretized streams, which enable efficient batch
processing in parallel architectures, is the enabling lower
level for PWW. 

PWW uses varying window sizes to accommodate for differences in
data. Another approach in batched stream processing is to use
\textit{adaptive window size}. Adaptive window algorithms is a
field of active research~\cite{ZLZ+04,BG07,BPR+13,YMD13}.
However, this research represents a different approach, in which
the window size is changed sequentially and adaptively, for future
windows based on earlier seen data. In PWW, several windows of
fixed sizes are applied in parallel, in a parameter-free manner
suitable for simple and robust implementation. Windows of
doubling size were proposed for processing data streams in
earlier work~\cite{AHJ03}, however the approach employed in PWW
is significantly different in that temporal windows of unbounded
doubling durations are applied in parallel, while still
ensuring efficient use of resources.

\section{Contribution and Future Research}
\label{sec:contribution}

This paper introduced the Progressive Window Widening algorithm
for data stream processing using temporal sliding windows. The
algorithm
\begin{itemize}
    \item solves the dilemma of smaller window size at a cost
        of inability to recognize longer patterns versus larger
        windows but slower response;
    \item works in parallel, in a manner suitable for modern
        multi-core multi-node cluster architectures;
    \item uses computational resources efficiently, imposing
        only a constant factor overhead compared to an algorithm
        based on a single window size.
\end{itemize}

The basic algorithm described in the paper brings a solution to
the stated problem. At the same time, the algorithm design poses
a number of questions and opens several research directions. 
\begin{itemize}
    \item Many adaptive window algorithms are, unlike PWW, essentially
        sequential. Modern data frameworks provide an opportunity to
        exploit the parallelism for more flexible and efficient
        adaptation.
    \item Doubling of batch durations is chosen in PWW due to
        simplicity of implementation and analysis. A different
        allocation of window sizes, either data-independent or
        adaptive, may bring better theoretical performance and
        practical results.
    \item PWW relies on batched stream processing,
        however it is only loosely coupled with the underlying 
        computing architecture, which is both an advantage and a
        drawback. A tighter coupling with lower-level
        stream processing may be helpful.
\end{itemize}
Along with others, these directions are deemed to be important
for future research.

\bibliographystyle{splncs03}
\bibliography{refs}

\clearpage
\section*{Appendix: Algorithm Implementations}

For real-life applications, the algorithm must be implemented
within a stream-processing framework, and different frameworks
provide different means and conveniences. For illustration, we
describe an implementation for {\it Apache Spark} \cite{ZCF+10}. We
provide code snippets in Scala and Python.

{\it Spark Streaming} implies that the stream processing
structure is defined statically rather than dynamically. Because
of that, all hierarchically combined streams should be defined
upfront.  Here comes handy the upper bound on the session
duration --- $T_{max}$. If we start with batch duration of 1
unit, and allocate $\lceil \log_2 T_{max} \rceil$ levels of
streams of combined batches, each session will fall entirely
within a sliding window at some level.

\begin{figure}
    \centering
	\includegraphics[scale=1.5]{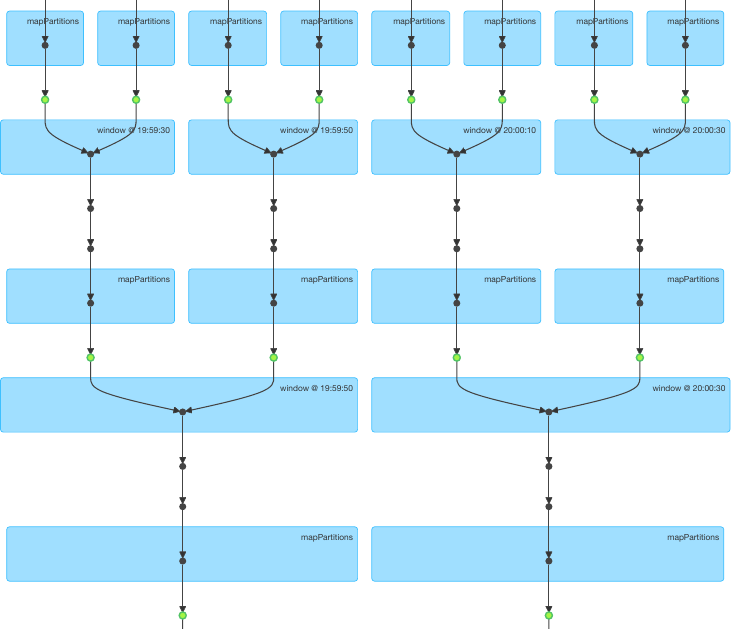}
	\caption{Progressive window widening in Apache Spark.}
    \label{fig:pww-spark}
\end{figure}

Code fragments illustrating an implementation of progressive window
widening are provided below. The code snippets
are also available as a
\href{https://gist.github.com/dtolpin/84158adce3e4218af06453771cae15f2}{\underline{\textbf{GitHub} Gist}}
\href{https://gist.github.com/dtolpin/84158adce3e4218af06453771cae15f2}{(http://tinyurl.com/hqugoyb)}.
A visualization of a Spark Streaming job executing progressive
window widening, as displayed by Apache Spark's web UI, is shown in
Figure~\ref{fig:pww-spark}.

\clearpage
\subsection*{Scala}

\lstdefinestyle{scalastyle}{
    frame=single,
    basicstyle=\linespread{1.2}\small\ttfamily,
	keywordstyle=\bfseries\color{green!60!black},
	commentstyle=\itshape\color{blue!40!gray},
	showstringspaces=false
}

The main loop is initialized with a stream  of batches of unit size.  Function
\texttt{detect} is called at each level, applies a pattern recognition
algorithm, and stores the result as a side effect.

\begin{lstlisting}[language=scala, style=scalastyle]
(1 to config.depth).foldLeft((batches, 1)) {
  case ((batch, batch_size), _) => {
    // Generate sliding windows with half-window step
    val windows = batches
      .window(Seconds(2*window_size), Seconds(window_size))
      .reduceByKey(_ ++ _)

    // Apply data mining/pattern recognition algorithm
    detect(windows)

    widen(batch, batch_duration, config.max_length)
  }
}
\end{lstlisting}

Functions \texttt{widen} and \texttt{combine} are defined as follows:

\begin{lstlisting}[language=scala, style=scalastyle]
def combine[A](a: Vector[A], b: Vector[A], max_length: Int)
    = {
  val ab = a ++ b
  if(ab.length > 2*max_length )
    ab.patch(max_length, Seq(), ab.length - 2*max_length);
  else
    ab
}

def widen(_batches: DStream[(String, Vector[Syscall])],
          _batch_duration: Int,
          max_length: Int) = {
  // Double batch duration
  val batch_duration = _batch_duration*2
  val batches = _batches
    .window(Seconds(batch_duration), Seconds(batch_duration))
    .reduceByKey(combine(_, _, max_length))
  (batches, batch_duration)
}
\end{lstlisting}

\clearpage
\subsection*{Python}

\lstdefinestyle{pythonstyle}{
    frame=single,
    basicstyle=\linespread{1.2}\small\ttfamily,
	keywordstyle=\bfseries\color{green!60!black},
	commentstyle=\itshape\color{blue!40!gray},
	showstringspaces=false
}

As in the Scala version, the main loop is initialized with a stream
of batches of unit size.  Function \texttt{detect} is called at
each level, applies a pattern recognition algorithm, and
stores the result as a side effect.

\begin{lstlisting}[language=python, style=pythonstyle]
t = 1
for _ in range(ceil(log2(max_time))):
    # Generate sliding windows with half-window step
    windows = (batches
        .window(2*t, t)
        .reduce(lambda a, b: a + b))

    # Apply data mining/pattern recognition algorithm
    detect(windows)

    # Double batch duration
    t *= 2
    batches = (batches
        .window_size(t, t)
        .reduce(lambda a, b: combine(a, b, max_length)))
\end{lstlisting}

Function \texttt{combine} is defined as follows:

\begin{lstlisting}[language=python, style=pythonstyle]
def combine(a, b, max_length):
    ab = a + b
    if len(ab) - max_lenbgth > max_length:
        ab[max_length:len(ab) - max_length] = []
    return ab
\end{lstlisting}

\end{document}